\documentclass[10pt,twocolumn,journal]{IEEEtran}
\usepackage{cite}
\usepackage[final]{graphicx}
\usepackage{amsmath}
\usepackage{times}
\usepackage{latexsym}
\usepackage{bm}
\usepackage{amssymb}
\usepackage[center]{caption2}
\usepackage{stfloats}
\usepackage{cases}
\usepackage{array}
\usepackage{setspace}
\usepackage{fancyhdr}
\usepackage{algorithm}
\usepackage{float}
\usepackage{color}

\usepackage{flushend}

\graphicspath{{figures/}}

\newcounter{mytempeqncnt}
\allowdisplaybreaks[4]

\newcaptionstyle{mystyle1}{%
  \centering TABLE  \captiontext \par}
\captionstyle{mystyle1}
\newcaptionstyle{mystyle2}{%
  \captionlabel $.$ \: \captiontext \par}
\captionstyle{mystyle2}
\newcaptionstyle{mystyle3}{%
  \captionlabel $.$ \: \captiontext \par}
\captionstyle{mystyle3}

\newtheorem{defn}{Definition}
\newtheorem{thm}{Theorem}
\newtheorem{lem}{Lemma}

\begin{document}

\title{Power Control via Stackelberg Game \\ for Small-Cell Networks}
\author{Yanxiang~Jiang,~\IEEEmembership{Senior~Member,~IEEE},
Hui~Ge,
Mehdi~Bennis,~\IEEEmembership{Senior~Member,~IEEE}, ~Fu-Chun~Zheng,~\IEEEmembership{Senior~Member,~IEEE} and~Xiaohu~You,~\IEEEmembership{Fellow,~IEEE}
\thanks{Manuscript received November 7, 2018, revised February 22, 2019, and accepted \today.}
\thanks{Y. Jiang is with the National Mobile Communications Research Laboratory, Southeast University, Nanjing 210096, China,
the State Key Laboratory of Integrated Services Networks,  Xidian University, Xi'an 710071, China,
and the Key Laboratory of Wireless Sensor Network $\&$ Communication, Shanghai Institute of Microsystem and Information Technology, Chinese Academy of Sciences, 
Shanghai 200050, China (e-mail: yxjiang@seu.edu.cn).}
\thanks{H. Ge and X. You are with the National Mobile Communications Research Laboratory, Southeast University, Nanjing 210096, China (e-mail: \{hge, xhyu\}@seu.edu.cn).}
\thanks{M. Bennis is with the Centre for Wireless Communications, University of Oulu, Oulu 90014, Finland (e-mail: mehdi.bennis@oulu.fi).}
\thanks{F. Zheng is with the School of Electronic and Information Engineering, Harbin Institute of Technology, Shenzhen 518055,  China,
and the National Mobile Communications Research Laboratory, Southeast University, Nanjing 210096, China (e-mail: fzheng@ieee.org).}
}

\maketitle


\begin{abstract}
In this paper, power control in the uplink for two-tier small-cell networks is investigated. We formulate the power control problem as a Stackelberg game, where the macrocell user equipment (MUE) acts as the leader and the small-cell user equipment (SUEs) act as the followers. To reduce the cross-tier and co-tier interferences and the power consumptions of both the MUE and SUEs, we propose optimizing not only the transmit rate but also the transmit power. The corresponding optimization problems are solved through a two-layer iteration. In the inner iteration, the SUEs compete with each other, and their optimal transmit powers are obtained through iterative computations. In the outer iteration, the optimal transmit power of the MUE is obtained in a closed form based on the transmit powers of the SUEs through proper mathematical manipulations. We prove the convergence of the proposed power control scheme, and we also theoretically show the existence and uniqueness of the Stackelberg equilibrium (SE) in the formulated Stackelberg game. The simulation results show that the proposed power control scheme provides considerable improvements, particularly for the  MUE.
\end{abstract}


\section{Introduction}

With the rapid development of the global communications industry, the problem of high energy consumption of communications systems is becoming increasingly more serious, and determining how to effectively improve the energy efficiency of the entire network is becoming increasingly more urgent. The introduction of small cells can greatly reduce  the energy consumption of the entire network. Furthermore, because small cells have small cell radii with small base stations (SBSs) deployed closer to  users and because short-distance transmissions have smaller path loss and fading compared to long-distance transmissions, the throughput of the entire network can be increased. 
Therefore, the energy efficiency of the entire two-tier small-cell network, which is composed of macrocells and a large number of small cells, can be greatly improved \cite{Hoydis11}. To improve the spectral efficiency, small cells can share the spectrum with macrocells; however, the co-tier and cross-tier interferences will seriously degrade the system performance due to the sharing of the spectrum. In this regard, proper power control in small-cell networks is required to reduce the interference and the power consumption.

Power control is an important research topic that has been widely investigated {in the literature \cite{Jiang01, Jiang02, Zhang02, Sawyer}}.
In two-tier small-cell networks, small cells can be deployed randomly and freely, and game theory has increasingly been used to achieve distributed power control \cite{Zhang, Jiang03}.
{In \cite{Yang}, the interference dynamics caused by time-varying environment was considered and a robust mean field game was proposed to control the transmit power of SBSs.}
In \cite{YA, Kang01, Li}, it was shown that a Stackelberg game can provide a suitable framework for modeling the competition in two-tier networks.
Specifically, a power control problem was formulated to maximize energy efficiency with minimal information exchange in \cite{YA}.
In \cite{Kang01}, both uniform and non-uniform pricing schemes were proposed to obtain the optimal resource allocation with a tolerable interference power constraint.
In \cite{Li}, a network interference controller was proposed to minimize the sum interference by pricing the power consumptions.
However, most of the existing literature only addressed  power control in the downlink and ignored power control of both  small-cell user equipment (SUEs) and macrocell user equipment  (MUE) in the uplink.\footnote{Note that the UE served by an SBS is called SUE, and the UE served by an MBS is called MUE.} Moreover, because most of the existing literature addressed power control through a price-based Stackelberg game, they can determine the optimal price and  power control for only one type of device. Therefore, if we formulate the power control problem through a Stackelberg game without pricing, the optimal power control for the two types of UEs in the considered two-tier small-cell networks can be determined simultaneously.

Motivated by the aforementioned discussions, we develop a power control scheme taking  both SUEs and MUE into account.  First, the power control problem of the considered two-tier small-cell networks is mathematically formulated as a Stackelberg game that consists of one leader and multiple followers. Second, the optimization problem is solved through a two-layer iteration. In the inner iteration, the followers compete with each other, and their optimal transmit powers are obtained through iterative computations. In the outer iteration, the optimal transmit power of the leader is calculated based on the transmit powers of the followers. Then, we theoretically show the convergence of the proposed scheme and the existence and uniqueness of the Stackelberg equilibrium (SE) in the Stackelberg game. Finally, our proposed power control scheme is verified through simulations, showing that it greatly improves the performance of the MUE.

The remainder of this paper is organized as follows. In Section \uppercase\expandafter{\romannumeral2}, the system model of the considered two-tier small-cell networks is presented. In Section \uppercase\expandafter{\romannumeral3}, the proposed  power control scheme via a Stackelberg game is developed. The simulation results are presented in Section \uppercase\expandafter{\romannumeral4}.
Final conclusions are drawn in Section \uppercase\expandafter{\romannumeral5}.

\section{System Model}

\begin{figure}[!b]
\centering 
\includegraphics[width=0.40 \textwidth]{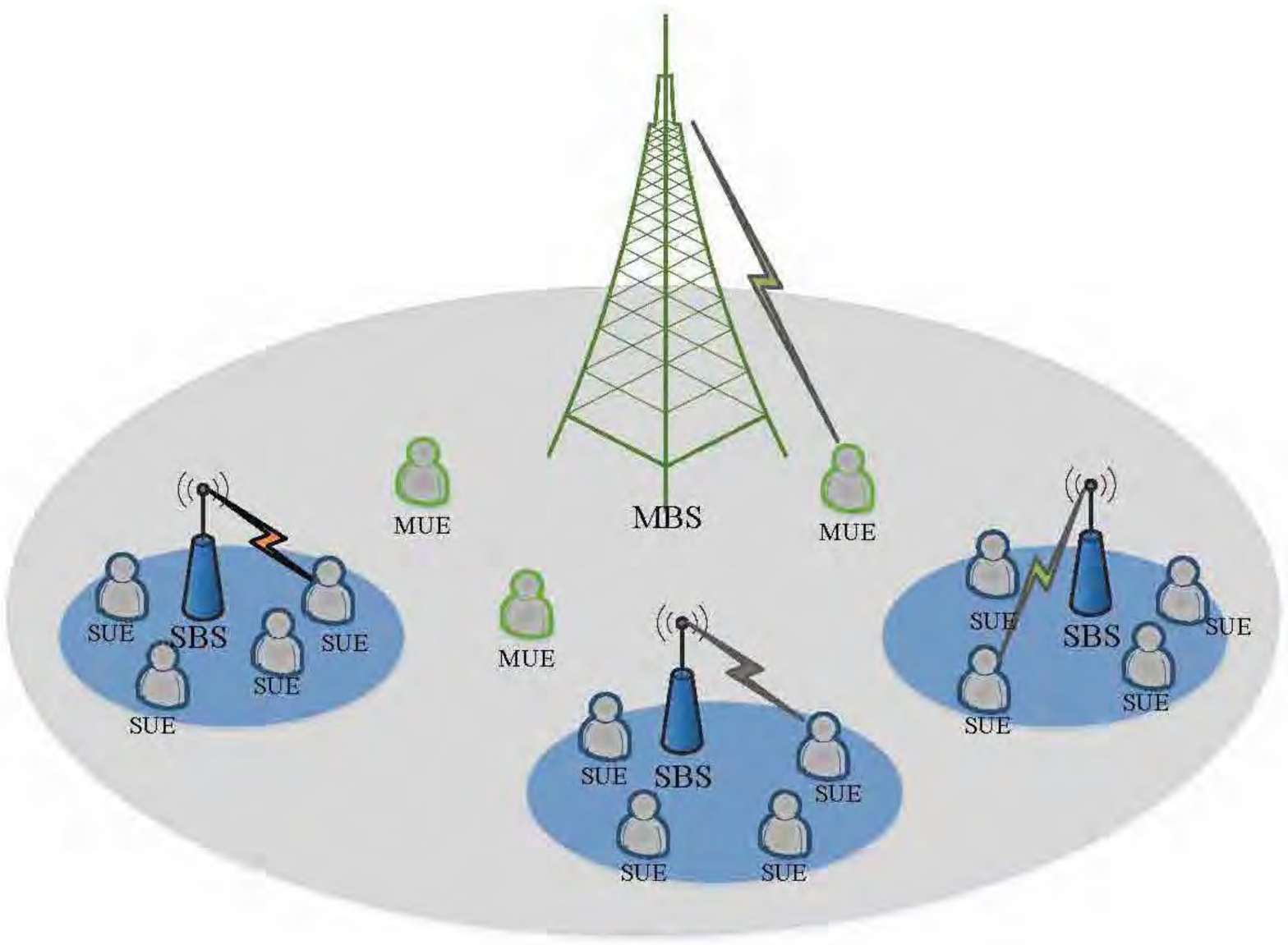}
\caption{The schematic of the considered two-tier small-cell network.}
\label{fig1}
\end{figure}

Consider the two-tier small-cell network shown in Fig. 1, which consists of one macrocell and $K$ small cells. Assume that the macro base station (MBS) and SBSs share the same spectrum and that
only one UE communicates with each BS at any time.
In the uplink, each SBS will experience interference from the MUE and its nearby SUEs, and the MBS will experience interference from its nearby SUEs.
Let $P_0$ denote the transmit power of the MUE served by the MBS, $P_k$ denote the transmit power of the $k$th SUE, and $\boldsymbol{p}=[P_1, P_2,\cdots, P_k, \cdots, P_K]^{T}$ denote the transmit power vector of the considered $K$ SUEs. Then, the transmit rate of the MUE served by the MBS can be expressed as follows:
\begin{equation}\label{eq2.1}
R_0 (P_0, \boldsymbol{p}) =  \ln \left( 1+ \frac{{H_{00}^{}{P_0}}}{{{N_0}{\rm{ + }}\sum\limits_{k = 1}^K {H_{0k}^{}P_k^{}} }} \right),
\end{equation}
where $H_{00}$ denotes the channel gain from the MUE to its corresponding MBS, $H_{0k}$ is the interference channel gain from the  $k$th SUE to the MBS, and $N_0$ is the noise power. The transmit rate of the SUE served by the $k$th SBS can be expressed as follows:
\begin{multline}\label{eq2.2}
{R _k}({P_k},{{\boldsymbol{p}}_{ - k}},{P_0})    \\
= \ln \left( 1+ \frac{{H_{kk}^{}{P_k}}}{{{N_0}{\rm{ + }}H_{k0}^{}P_0^{}{\rm{ + }}\sum\limits_{k' \ne k,k' = 1}^K {H_{kk'}^{}P_{k'}^{}} }} \right),
\end{multline}
where $\boldsymbol{p}_{-k}$ denotes the transmit power vector of the $K-1$ other SUEs and $\boldsymbol{p}_{-k}=[P_1, P_2,\cdots, P_{k-1},$ $ P_{k+1},\cdots, P_K]$, $H_{kk}$ is the channel gain from the $k$th SUE to its corresponding SBS, $H_{k0}$ is the interference channel gain from the MUE to the $k$th  SBS, and $H_{kk'}$ is the interference channel gain from the $k'$th SUE to the $k$th SBS.

The design objective of this paper is to develop a power control scheme that can increase the transmit rate  with reduced co-tier and cross-tier interferences and power consumption. Moreover, this paper aims to achieve the above design objective for  two-tier small-cell networks where there are two types of UEs, i.e., MUE and SUE, and two different cell types, i.e., macro cell and small cell.

\section{The Proposed Power Control Scheme via Stackelberg Game}

In this section, we propose a power control scheme for two-tier small-cell networks based on a Stackelberg game, which has one leader and multiple followers.
In the formulated Stackelberg game, the MUE, acting as the leader, is supposed to make its own decision and maximize its utility with the best responses of the followers, and the SUEs acting as the followers will respond to the leader's action and maximize their utilities through a subgame \cite{Fudenberg, Bennis02, Bennis03}.
Note that the transmit power of the MUE or SUE is controlled by its corresponding MBS or SBS and the MBS can 
control its corresponding SBSs 
in the considered two-tier small-cell network.
{When the transmit power of the MUE has been determined by the MBS, this transmit power information will be sent from the MBS to its corresponding SBSs.}
Therefore, the MUE controlled by the MBS acts as the leader, and the SUEs controlled by its corresponding SBSs act as the followers.

\subsection{Stackelberg Game Formulation}

From \eqref{eq2.1} and \eqref{eq2.2}, we find that the transmit rate of the MUE can be improved by increasing the transmit power of the MUE but at the cost of increased cross-tier interference to the SUEs. Likewise,  the transmit rate of the SUEs can be improved by increasing the transmit power of the corresponding SUE but at the cost of increased cross-tier interference to the MUE and increased co-tier interference to the other $K-1$ SUEs.
To reduce the cross-tier and co-tier interferences and the power consumptions of both the MUE and SUEs, we propose optimizing not only the transmit rate but also the power consumption.
First, the leader MUE moves and determines its transmit power.
Subsequently, the follower SUEs move and update their power control strategies to maximize their individual utilities based on the MUE's transmit power.

We define the utility function of the MUE as follows:
\begin{equation}\label{eqA-3}
{U_{{0}}}({P_0},{\boldsymbol{p}}) =  R_0 (P_0, \boldsymbol{p})  - {\lambda _0}P_0^{},
\end{equation}
where $\lambda_0$ denotes the coefficient characterizing the influence of per unit transmission power for MUE \cite{YLi}.
Then, the optimization problem of the MUE can be expressed as follows:
\begin{equation}\label{eq1}
\begin{array}{l}
\mathop {\max }\limits_{P_0} {U_{{0}}}({P_0},{\boldsymbol{p}}) ,\\
\text{s.t.} \;\; 0 \le {P_0} \le {P_T},  
\end{array}
\end{equation}
where $P_T$ denotes the maximum transmit power of the MUE or SUEs.

We define the utility function of the $k$th SUE as follows:
\begin{equation}
{U_k}({P_k},{{\boldsymbol{p}}_{ - k}},{P_0}) = R_k ({P_k},{{\boldsymbol{p}}_{ - k}},{P_0}) - {\lambda _k}P_k^{},
\end{equation}
where $\lambda_k$ denotes the coefficient characterizing the influence of per unit transmission power for SUE.
Then, the optimization problem of the $k$th SUE can be expressed as follows:
\begin{equation}\label{eq2}
\begin{array}{l}
\mathop {\max }\limits_{{P_k}} {U_k}({P_k},{{\boldsymbol{p}}_{ - k}},{P_0}), \\
\text{s.t.}\;\; 0 \le {P_k} \le {P_T},\forall k \in \{ 1,2,\cdots,K \}.
\end{array}
\end{equation}

The optimization problems in \eqref{eq1} and \eqref{eq2} lead to a Stackelberg game. In this game, the objective is to find the SE point from which neither the leader nor the followers have incentives to deviate.
{Just similar to the definition in \cite{Kang01}, we define the SE as follows.}
\begin{defn}
Let $P_0^*$ and $P_k^{*}$ denote the two solutions for the optimization problems in \eqref{eq1} and \eqref{eq2}, respectively. Let $\boldsymbol{p}^* = [P_1^{*}, P_2^{*}, \cdots, P_k^{*}, \cdots, P_K^{*}]^T$ and $\boldsymbol{p}_{-k}^* = [P_1^{*}, P_2^{*}, \cdots, P_{k-1}^{*}, P_{k+1}^{*}, \cdots, $ $P_K^{*}]^T$. Then, $(P_0^*, \boldsymbol{p}^*)$ is an SE point for the proposed Stackelberg game if the following conditions are satisfied:
\begin{equation}
{U_{{0}}}(P_0^*,{{\boldsymbol{p}}^*}) \ge {U_{{0}}}(P_0^{},{{\boldsymbol{p}}^*}),
\end{equation}
\begin{equation}
{U_k}(P_k^*,{\boldsymbol{p}}_{ - k}^*,P_0^*) \ge {U_k}(P_k^{},{\boldsymbol{p}}_{ - k}^*,P_0^*).
\end{equation}
\end{defn}

{Generally, the SE for a Stackelberg game can be obtained by finding its subgame perfect Nash equilibrium (NE) \cite{Kang01, Gibbons}.
In our proposed Stackelberg game, it can be readily seen that the SUEs compete in a non-cooperative fashion.
Therefore, a non-cooperative power control subgame is formulated, where the corresponding NE is defined as the operating point at which no player can improve utility by changing its strategy unilaterally \cite{Kang01}.}

To obtain the SE of the proposed Stackelberg game, we propose exploiting the backward induction method \cite{ZSu} to solve the above optimization problems.
Generally, the followers' best responses can be obtained with the fixed value given by the leader, and then the optimal strategy of the leader can be achieved according to the followers' best responses.
Correspondingly, we can first solve the followers' optimization problem in \eqref{eq2}. Then, by using the obtained solution, we can solve the leader's optimization problem in \eqref{eq1}.

\subsection{The Optimal Solution of the Followers' Optimization Problem}

\begin{figure*}[!b]
\normalsize
\setcounter{mytempeqncnt}{\value{equation}}
\setcounter{equation}{14}
\hrulefill
 {\setlength\arraycolsep{2pt}}
\begin{equation}\label{eqnew18}
P_{0,k}^ {\rm{min}}  =
   \left\{ {\begin{array}{*{20}{l}}
   {{0},} \hfill & {P_k^{\rm{tmp}} > {P_T},} \hfill  \\
   \mathop {\max }\limits  \left\{ 0, \frac{{\left(\frac{1}{{{\lambda _k}}} - {P_T}\right)H_{kk}^{} - {N_0} - \sum\limits_{k' \ne k , k' = 1  }^K {H_{kk'}^{}\tilde P_{k'}^*} }}{{H_{k0}^{}}} \right\}, \hfill &  {0 < P_k^{\rm{tmp}}\le P_T,} \hfill  \\
   \mathop {\max }\limits  \left\{ 0, {\frac{{\frac{1}{{{\lambda _k}}} H_{kk}^{} - {N_0} - \sum\limits_{k' \ne k,k' = 1}^K {H_{kk'}^{}\tilde P_{k'}^*} }}{{H_{k0}^{}}}}\right\}, \hfill &  {P_k^{\rm{tmp}}\le 0,} \hfill  \\
\end{array}} \right. \\
\end{equation}
\begin{equation}\label{eqnew19}
P_{0,k}^ {\rm{max}}  =
   \left\{ {\begin{array}{*{20}{l}}
   {\mathop {\min }\limits  \left\{ P_T, {\frac{{\left(\frac{1}{{{\lambda _k}}} - {P_T}\right)H_{kk}^{} - {N_0} - \sum\limits_{k' \ne k,k' = 1}^K {H_{kk'}^{}\tilde P_{k'}^*} }}{{H_{k0}^{}}}} \right\}, \hfill}  &  {P_k^{\rm{tmp}} > {P_T},} \hfill  \\
   \mathop {\min }\limits  \left\{ P_T, {\frac{{\frac{1}{{{\lambda _k}}} H_{kk}^{} - {N_0} - \sum\limits_{k' \ne k,k' = 1}^K {H_{kk'}^{}\tilde P_{k'}^*} }}{{H_{k0}^{}}}} \right\}, \hfill &  {0 < P_k^{\rm{tmp}} \le P_T,} \hfill  \\
   P_T,  & {P_k^{\rm{tmp}} \le 0.} \hfill  \\
\end{array}} \right.
\end{equation}
\setcounter{equation}{\value{mytempeqncnt}}
\end{figure*}

We have the following theorem for the optimal solution of the optimization problem in  \eqref{eq2} for the followers.
\begin{thm}\label{thm1}
Given the transmit power of the MUE, the optimization problem in  \eqref{eq2} has a globally optimal solution, as follows:
\begin{equation}\label{eq3}
\tilde P_k^ *  =
   \left\{ {\begin{array}{*{20}{l}}
   {{P_T},} \hfill & {P_k^{\rm{tmp}} > {P_T},} \hfill  \\
   P_k^{\rm{tmp}}, \hfill & {0 < P_k^{\rm{tmp}} \le P_T,}   \\
   0,  & {P_k^{\rm{tmp}} \le 0,}
\end{array}} \right.
\end{equation}
or
\begin{equation}\label{eqB-10}
\tilde P_k^ *  = P_T - \left[ P_T - (P_k^{\rm{tmp}})^+ \right] ^ +,
\end{equation}
where  ${(\cdot)^{{ + }}} \buildrel \Delta \over = \max (\cdot, 0) $,
\begin{equation}\label{eq11}
P_k^{\rm{tmp}} = \frac{1}{{{\lambda _k}}} - \frac{{{N_0}{\rm{ + }}H_{k0}^{}P_0^{}{\rm{ + }}\sum\limits_{k' \ne k,k' = 1}^K {H_{kk'}^{}P_{k'}} }}{{H_{kk}^{}}},
\end{equation}
and $P_k^{\rm{tmp}}$ denotes the temporary value of the optimal transmit power of the $k$th SUE.
\end{thm}
\begin{proof}
As shown, the utility function ${U_k}({P_k},{{\boldsymbol{p}}_{ - k}},{P_0})$  is strictly concave. Furthermore, it can be verified that the SUEs' strategy space is a non-empty and close-bounded convex set in Euclidean space.
Correspondingly, the optimization problem in \eqref{eq2} can readily be proven to be convex; thus, it has a globally optimal solution. By setting the first-order derivative of $U_k({P_k},{{\boldsymbol{p}}_{ - k}},{P_0})$ with respect to $P_k$ to zero, {$P_k^{\rm{tmp}}$} can readily be calculated as shown in \eqref{eq11}.
By considering the constraint $0 \le {P_k} \le {P_T}$, the optimal solution of the optimization problem in \eqref{eq2} can readily be obtained as shown in \eqref{eq3} or \eqref{eqB-10}.
This completes the proof.
\end{proof}

\subsection{The Optimal Solution of the Leader's Optimization Problem}

\begin{figure*}[!b]
\normalsize
\setcounter{mytempeqncnt}{\value{equation}}
\setcounter{equation}{20}
 {\setlength\arraycolsep{2pt}}
\begin{equation}\label{eq5}
\tilde P_0^ *  =
\left\{ {\begin{array}{*{20}{l}}
   \arg \max \left\{ {U_{{0}}}(P_0^{\max }, \boldsymbol{\tilde p}^*),{U_{{0}}}(P_0^{\min} , \boldsymbol{\tilde p}^*) \right\} , &  {C_2}^2 - 4{C_1}{C_3} < 0, \\
   \arg \max \left\{ {U_{{0}}}(P_0^{\max }, \boldsymbol{\tilde p}^*),{U_{{0}}}(P_0^{\min}, \boldsymbol{\tilde p}^*), {U_{{0}}}(P_0^1, \boldsymbol{\tilde p}^*),{U_{{0}}}(P_0^2, \boldsymbol{\tilde p}^*) \right\}, & {C_2}^2 - 4{C_1}{C_3} \ge 0,  \\
\end{array}} \right.
\end{equation}
\setcounter{equation}{\value{mytempeqncnt}}
\end{figure*}

After some mathematical manipulations, we can obtain the refined constraint for $P_0$ according to \eqref{eq3} as follows:
\begin{equation}
P_0^ {\rm{min}}  \le P_0 \le P_0^ {\rm{max}} ,
\end{equation}
where
\begin{align}
P_{0}^ {\rm{min}} & =  \mathop {\max } _ { k \in \left\{ 1, 2, \cdots, K \right\}}\limits   P_{0,k}^ {\rm{min}}, \\
P_{0}^ {\rm{max}} & =  \mathop {\min } _ { k \in \left\{ 1, 2, \cdots, K \right\}}\limits   P_{0,k}^ {\rm{max}},
\end{align}
and $P_{0,k}^ {\rm{min}}$ and $P_{0,k}^ {\rm{max}} $ are shown at the bottom of next page in \eqref{eqnew18} and \eqref{eqnew19}, respectively.
\stepcounter{equation}\stepcounter{equation}

Define $\boldsymbol{\tilde p}^* = [\tilde P_1^{*}, \tilde P_2^{*}, \cdots, \tilde P_k^{*}, \cdots, \tilde P_K^{*}]^T$. Substituting \eqref{eqB-10} into \eqref{eqA-3}, and after some mathematical manipulations, we obtain
\begin{multline}
U_0 (P_0, \boldsymbol{ \tilde  p}^*) =  \ln \left( 1+ \frac{{H_{00}^{}{P_0}}}{{{N_0}{\rm{ + }}\sum\limits_{k = 1}^K {H_{0k}^{}\tilde P_k^{*}} }} \right)  - {\lambda _0}P_0, \\
 = \ln \left( 1+ \frac{{H_{00}^{}{P_0}}}{{{N_0}{\rm{ + }}\sum\limits_{k = 1}^K {H_{0k}^{}
\left\{ P_T - \left[ P_T - (P_k^{\rm{tmp}})^+ \right] ^ + \right\} }  }} \right)  \\
\hspace*{170pt} - {\lambda _0}P_0, \\
 = \ln \left( 1+ \frac{{H_{00}^{}{P_0}}}{{{N_0}{\rm{ + }}\sum\limits_{k = 1}^K {H_{0k}^{}
\left[ P_T - \varepsilon_k ' \cdot ( P_T - \varepsilon_k \cdot P_k^{\rm{tmp}} )  \right] }  }} \right) \\
 - {\lambda _0}P_0,
\end{multline}
where $\varepsilon_k$ denotes the indicator function with $\varepsilon_k =1$ if $P_k^{\rm{tmp}} > 0 $ and  $\varepsilon_k =0$ otherwise, and $\varepsilon_k '$ is the indicator function with $\varepsilon_k ' =1$ if $P_T - (P_k^{\rm{tmp}})^+ > 0 $ and  $\varepsilon_k ' =0$ otherwise.
After  some further manipulations,
the optimization problem for the leader  in \eqref{eq1} can be reformulated as follows:
\begin{equation}\label{eq4}
\begin{array}{l}
\mathop {\max }\limits_{P_0}  U_{0}({P_0},{\boldsymbol{\tilde p}^*}) = \mathop {\max }\limits_{P_0} \left\{ {\ln} \left(1 + \frac{{H_{00}^{}P_0^{}}}{{A - BP_0^{}}} \right) - {\lambda _0}P_0^{} \right\},\\
\text{s.t.}\;\; P_0^ {\rm{min}} \le P_0 \le P_0^ {\rm{max}},
\end{array}
\end{equation}
where
\begin{align}
A  = &{N_0}{\rm{ + }}\sum\limits_{k = 1}^K {H_{0k}^{} \left[ P_T - \varepsilon_k '  P_T \right. } \nonumber \\
 & {\left. + \varepsilon_k ' \varepsilon_k \left( \frac{1}{{{\lambda _k}}} - \frac{{{N_0}{\rm{ + }} \sum\limits_{k' \ne k,k' = 1}^K {H_{kk'}^{} \tilde P_{k'}^*} }}{{H_{kk}^{}}}   \right)  \right] } ,\\
B  = & \sum\limits_{k = 1}^K { \varepsilon_k ' \varepsilon_k \frac{{H_{0k}^{}H_{k0}^{}}}{{H_{kk}^{}}}}.
\end{align}

Then, we have the following theorem.
\begin{thm}\label{thm2}
If $\mathop {\sum }\limits_{k=1}^K \varepsilon_k   \varepsilon_k ' \ne 0$, then the optimization problem in \eqref{eq4} has an optimal solution as shown in \eqref{eq5} at the bottom of this page, where  \stepcounter{equation}
\begin{align}
C_1 &= {\lambda _0}B({H_{00}} - B),\\
{C_2} &= {\lambda _0}A(2B - {H_{00}}),\\
{C_3} &= A{H_{00}} - {\lambda _0}{A^2},\\
P_0^{1} &= \frac{{ - {C_2} + \sqrt {{C_2}^2 - 4{C_1}{C_3}} }}{{2{C_1}}},\\
P_0^{2} &= \frac{{- {C_2} - \sqrt {{C_2}^2 - 4{C_1}{C_3}} }}{{2{C_1}}}.
\end{align}
\end{thm}
\begin{proof}
If $\mathop {\sum }\limits_{k=1}^K \varepsilon_k   \varepsilon_k ' \ne 0$, then $B \ne 0$.
Take the first-order derivative of $U_{0}({P_0},{\boldsymbol{\tilde  p}^*})$ with respect to $P_0$. Then, we have
\begin{multline}
\frac{{\partial {U_{{0}}({P_0},{\boldsymbol{\tilde p}^*})}}}{{\partial {P_0}}} \\
= \frac{{AH_{00}^{}}}{{{{(A - BP_0^{})}^2} + H_{00}^{}P_0^{}(A - BP_0^{})}} - {\lambda _0}.
\end{multline}
Set the above expression to zero. Then, we have
\begin{multline}\label{eq6}
 {\lambda _0}B(H_{00}^{} - B)P_0^2 + {\lambda _0}A(2B - H_{00}^{})P_0^{} + AH_{00}^{} - {\lambda _0}{A^2}  \\
 = C_1 P_0^2 + C_2 P_0 + C_3 = 0.
\end{multline}
If ${C_2}^2 - 4{C_1}{C_3} < 0$, then \eqref{eq6} has no solution. Correspondingly, the objective function of the optimization problem in \eqref{eq4} is definitely a monotonic function, and its solution must be  one of the two endpoints. Then, we have
\begin{equation}
\tilde P_0^ *  = \arg \max \left\{ {U_{{0}}}(P_0^{\max }, \boldsymbol{\tilde p}^*),{U_{{0}}}(P_0^{\min} , \boldsymbol{\tilde p}^*) \right\} .
\end{equation}
If ${C_2}^2 - 4{C_1}{C_3} \ge 0$, then we can obtain the two solutions of \eqref{eq6}, i.e., $P_{0}^{ 1}, P_{0}^{ 2}$.
Since the objective function of the optimization problem in \eqref{eq4} is a continuous function, its solution must be among the extreme points and the endpoints. Correspondingly, we have
\begin{multline}
\tilde P_0^ *  =
\arg \max \left\{ {U_{{0}}}(P_0^{\max }, \boldsymbol{\tilde p}^*),{U_{{0}}}(P_0^{\min}, \boldsymbol{\tilde p}^*), \right. \\
\left. {U_{{0}}}(P_0^1, \boldsymbol{\tilde p}^*),{U_{{0}}}(P_0^2, \boldsymbol{\tilde p}^*) \right\}.
\end{multline}

This completes the proof.
\end{proof}

\subsection{The Proposed Power Control Scheme via Stackelberg Game}

\begin{algorithm}[!t]
\caption{The Proposed Power Control Scheme via Stackelberg Game}
\begin{itemize}
\item Step 1: Initialization: $m=1, n=1$, $P_0(1)$, $\hat{P}_k(1)$ for $ 1 \le k \le K$.
\item Step 2: Update $\hat{P}_k(m)$ as follows:
\begin{multline*}
\hat P_k^{}\left( {m + 1} \right) \\
  ={\frac{1}{{{\lambda _k}}} - \frac{{{N_0}{\rm{ + }}H_{k0}^{}P_0(n){\rm{ + }}\sum\limits_{k' \ne k,k' = 1}^K {H_{kk'}^{}\hat{P}_{k'}(m)} }}{{H_{kk}^{}}}} ,
\end{multline*}
and set  $m=m+1$.
\item Step 3: Repeat step $2 $ until the inner iteration converges.
\item Step 4: According to \eqref{eqB-10}, calculate  the  transmit power of  each SUE, $\tilde P_k^ *$, as follows:
\begin{equation*}
\tilde P_k^ *  =
   P_T - \left\{ P_T - \left[\hat P_k(m)\right]^+ \right\} ^ +.
\end{equation*}
\item Step 5: According to \eqref{eq5}, calculate the  transmit power of the MUE, $\tilde P_0^*(n+1)$,
and set $n=n+1$.
\item Step 6: Repeat steps $2 \thicksim 5$ until the outer iteration converges.
\end{itemize}
\end{algorithm}
We are now ready to develop the proposed power control scheme based on the Stackelberg game  described in Algorithm 1.\footnote{Note that our proposed scheme can always achieve the optimal solution for any initial point.
On the one hand, we analyze theoretically in Section III-D that the convergence of the proposed scheme can always be guaranteed.
On the other hand, we will prove in Section III-E that one and only one SE point exists for the proposed Stackelberg game.}
In the proposed scheme, the MUE acts as the leader, the SUEs act as the followers, and the Stackelberg game is formed through the two-layer iteration. In the inner iteration, the SUEs compete with each other, and their own transmit powers are updated iteratively based on the transmit power of the MUE, as shown in Theorem \ref{thm1}.
In the outer iteration, the MUE updates its own transmit power based on the transmit powers of the SUEs, as shown in Theorem \ref{thm2}. In the proposed power control scheme, each user plays the best response strategy and maximizes its own utility function in each iteration given the chosen transmit powers of the other users in the previous iteration.

Let $\boldsymbol{W}$ denote a $K \times K$ matrix whose elements are given by
\begin{equation}
{{\boldsymbol{W}}_{kk'}}  =
\left\{ \begin{array}{*{20}{l}}
\frac{{{H_{kk'}}}}{{{H_{kk}}}}, & k \ne k', 1 \le k, k' \le K,\\
0, & k = k', 1 \le k, k' \le K.
\end{array} \right.
\end{equation}
Then, we can establish the following theorem for the convergence of the inner iteration of the proposed scheme.
\begin{thm}\label{thm3}
If the matrix norm of $\boldsymbol{W}$ is not larger than 1, i.e., $\left\| {\boldsymbol{W}} \right\| \le 1$, then the inner iteration of the proposed power control scheme via a Stackelberg game as shown in Algorithm 1 converges.
\end{thm}

\begin{proof}
Define
\begin{align}
{\phi _k} = & P_k^{\text{tmp}}, \\
\boldsymbol{\phi}(\boldsymbol{p}) = & [\phi _1, \phi _2, \cdots,\phi _k, \cdots, \phi _K]^T, \\
\boldsymbol{\mu}   = &{\left[ {\frac{1}{{{\lambda _1}}},\frac{1}{{{\lambda _2}}},\cdots, \frac{1}{{{\lambda _k}}},\cdots,\frac{1}{{{\lambda _K}}}} \right]^T},\\
\boldsymbol{\nu}   = &\left[ \frac{{{N_0} + {H_{10}}{P_0}}}{{{H_{11}}}},\frac{{{N_0} + {H_{20}}{P_0}}}{{{H_{22}}}}, \cdots, \right. \nonumber\\
 &\hspace*{20pt} \left. \frac{{{N_0} + {H_{k0}}{P_0}}}{{{H_{kk}}}}, \cdots, \frac{{{N_0} + {H_{K0}}{P_0}}}{{{H_{KK}}}} \right]^T.
\end{align}
Then, $\boldsymbol{\phi}(\boldsymbol{p})$ can be expressed in a vector-matrix form as follows:
\begin{equation}
\boldsymbol{\phi}(\boldsymbol{p}) =  \boldsymbol{\mu} - \boldsymbol{\nu}  - \boldsymbol{W} \boldsymbol{p}.
\end{equation}

Assume that $\left\| {\boldsymbol{W}} \right\| \le 1$.  Then, we can obtain the following relationship:
\begin{align}
 \left\| \boldsymbol{\phi}(\boldsymbol{p}) - \boldsymbol{\phi}(\boldsymbol{p}') \right\| &= \left\| {{\boldsymbol{W}}\left( {{\boldsymbol{p}} - {\boldsymbol{p}'}} \right)} \right\|\\
 &\le \left\| {\boldsymbol{W}} \right\| \cdot \left\| {{\boldsymbol{p}} - {\boldsymbol{p}'}} \right\|\\
 &\le \left\| {{\boldsymbol{p}} - {\boldsymbol{p}'}} \right\|
\end{align}
According to \cite{Shum}, we know that $\boldsymbol{\phi}(\boldsymbol{p})$ is a contraction.
Then, according to the Banach contraction theorem introduced in \cite{Shum}, $\boldsymbol{\phi}(\boldsymbol{p})$ has a unique fixed point that is globally asymptotically stable.
Correspondingly, the inner iteration of the proposed  power control scheme  via the Stackelberg game as shown in Algorithm 1 converges. This completes the proof.
\end{proof}

Note that we can always find a matrix norm of $\boldsymbol{W}$ to satisfy  $\left\| {\boldsymbol{W}} \right\| \le 1$. Therefore, the convergence of the inner iteration of the proposed power control scheme can always be guaranteed.

In the following, we briefly analyze the convergence of the outer iteration of the proposed scheme.
We know that the utility function $U_k$ is concave with respect to $p_k$. Therefore, the SUEs can gradually increase their transmit powers from an arbitrary small number to their optima.
Then, the obtained transmit power of the SUEs can be used to determine the  transmit power of the MUE.
When the transmit powers of the SUEs have been increased to their optima, the optimal transmit power of the MUE can then be determined accordingly  \cite{BWang}.
For the practical implementation of the proposed Stackelberg game, the SUEs can find their optimal transmit powers by gradually increasing the transmit power until the utility function $U_k$ reaches its maximum due to its concave property. Correspondingly, the MUE can always achieve its SE, i.e., the convergence of the outer iteration of the proposed scheme can always be guaranteed.

\subsection{The Existence and Uniqueness of the SE}

SE offers a predictable and stable outcome about the transmit power strategies that the MUE and each SUE will choose.
For the proposed Stackelberg game, we have the following theorem.
\begin{thm}
One and only one SE point exists for the proposed Stackelberg game.
\end{thm}
\begin{proof}
Generally, we can obtain the SE for the proposed Stackelberg game by finding the NE of its subgame.
For the proposed Stackelberg game, there is only one leader. Therefore, the best response of the leader  can readily be obtained by solving the optimization problem in \eqref{eq1}.
At the followers' side, the best response can be achieved by solving the optimization problem in \eqref{eq2}.
Correspondingly, to prove this theorem, we only need to prove that a unique NE point exists for the subgame at the followers' side.

It can be verified that the SUEs' strategy space is a non-empty and closed-bounded convex set in the Euclidean space. Moreover,
it can also be verified that the utility function ${U_k}({P_k},{{\boldsymbol{p}}_{ - k}},{P_0})$  is continuous with respect to  $P_{k}$.
In addition, the utility function ${U_k}({P_k},{{\boldsymbol{p}}_{ - k}},{P_0})$  is concave.
According to \cite{Rosen} and \cite{Alpcan}, we know that the NE exists if the players' strategy space is a non-empty and closed-bounded set in the Euclidean space and the utility function is continuous and concave in its strategy space.
Correspondingly, the existence of the NE of the subgame at the followers' side can be proved.


Regarding  the uniqueness of the NE, we first state the following lemma \cite{Beibei}.
\begin{lem}
 For a game, if its feasible region is convex and each players' utility function is strictly convex, then the NE of the game is unique.
\end{lem}

Then, according to the above-mentioned discussions and the proof of Theorem \ref{thm1}, we can easily verify the uniqueness of the subgame at the followers' side.

This completes the proof.
\end{proof}

\section{Simulation Results}

 In this section, the performance of the proposed power control scheme via a Stackelberg game is evaluated via simulations. In the simulations, the radii of the macrocell and small cells are set to be $1000$ m and $100$ m, respectively. The noise spectral density is set to $-174$ dBm/Hz. Unless otherwise stated, we set $\lambda= \lambda_0=\lambda_k, \forall k$. In the following, for description convenience, we use $\bar U_K, \bar R_K,$ and $\bar P_K$ to denote the average utility, the average transmit rate, and the average transmit power of the considered SUEs, respectively, and we use $\bar R$ to denote the average transmit rate of the considered MUE and SUEs.

In Fig. \ref{fig2} and Fig. \ref{fig3}, we show the utility of the MUE and the average utility of the SUEs of the proposed scheme versus the number of iterations with different $K$ for $\lambda=10^3$ and $P_{T}=0$ dBm. As shown, the proposed scheme
can converge to a stable state quickly, which verifies that the proposed scheme can converge to the SE.
Moreover, both the utility of the MUE and the average utility of the SUEs are observed to decrease with the increased number of the SUEs, which can be attributed to  the increased cross-tier or co-tier interference.

\begin{figure}[!t]
\centering 
\includegraphics[width=0.52\textwidth]{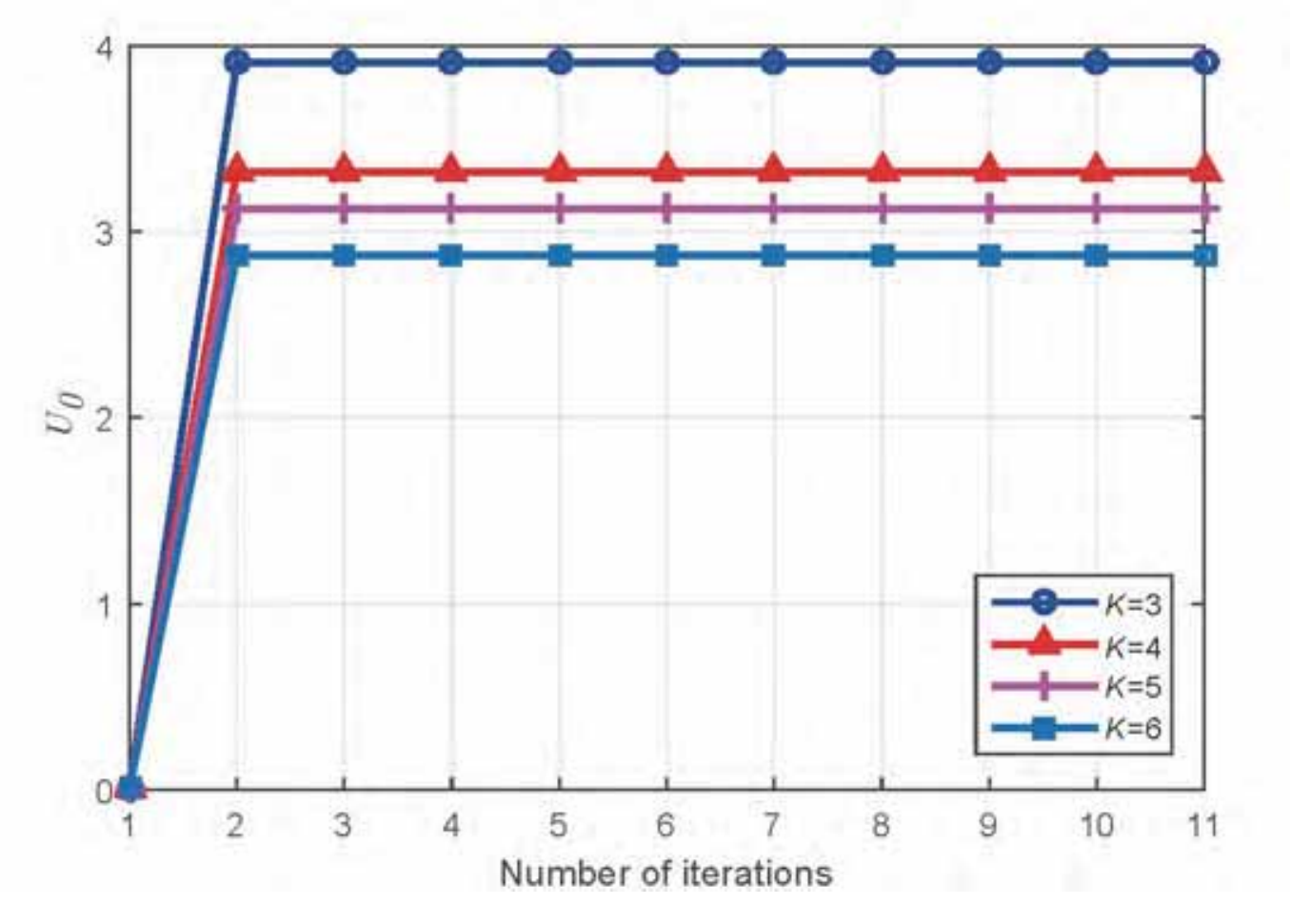}
\caption{The utility of the MUE versus the number of iterations with different $K$.}
\label{fig2}
\end{figure}
\begin{figure}[!t]
\centering 
\includegraphics[width=0.52\textwidth]{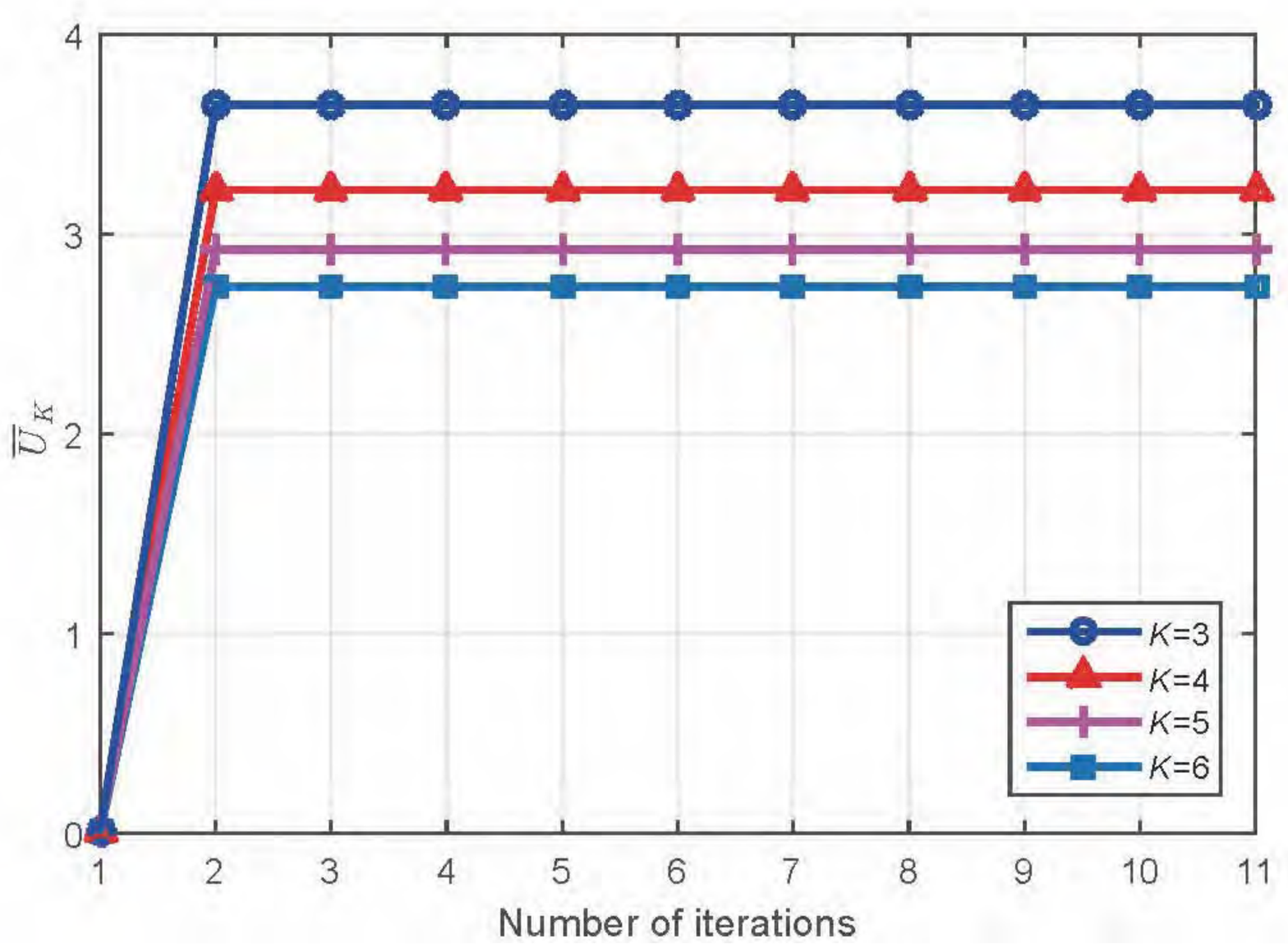}
\caption{The average utility of the SUEs versus the number of iterations with different $K$.}
\label{fig3}
\end{figure}

In Fig. \ref{fig4}, we show the performance comparison between our proposed scheme and the non-cooperative power control scheme in \cite{Vu} with $\lambda=10^3$, $P_{T}=0$ dBm, and $K=4$. For description convenience, the utility of the MUE and the average utility of the SUEs of the proposed Stackelberg-game-based power control scheme are referred to as SG-MUE and SG-SUE, respectively. The utility of the MUE and the average utility of the SUEs of the non-cooperative-game-based power control scheme in \cite{Vu} are referred to as NCG-MUE and NCG-SUE, respectively.
As shown, the SG-MUE is close to and slightly larger than the SG-SUE, and the NCG-MUE is approximately zero and clearly smaller than the NCG-SUE.
This result verifies that the proposed scheme can significantly improve the performance of the MUE.\footnote{Note here that the performance improvement is not absolutely free.
There is some system overhead between the leader and the followers in order to realize the Stackelberg game in our proposed scheme. However, the quick convergent property of our proposed scheme as illustrated in Fig. 2 and Fig. 3 indicates that the corresponding system overhead will be affordable.}

\begin{figure}[!t]
\centering 
\includegraphics[width=0.52\textwidth]{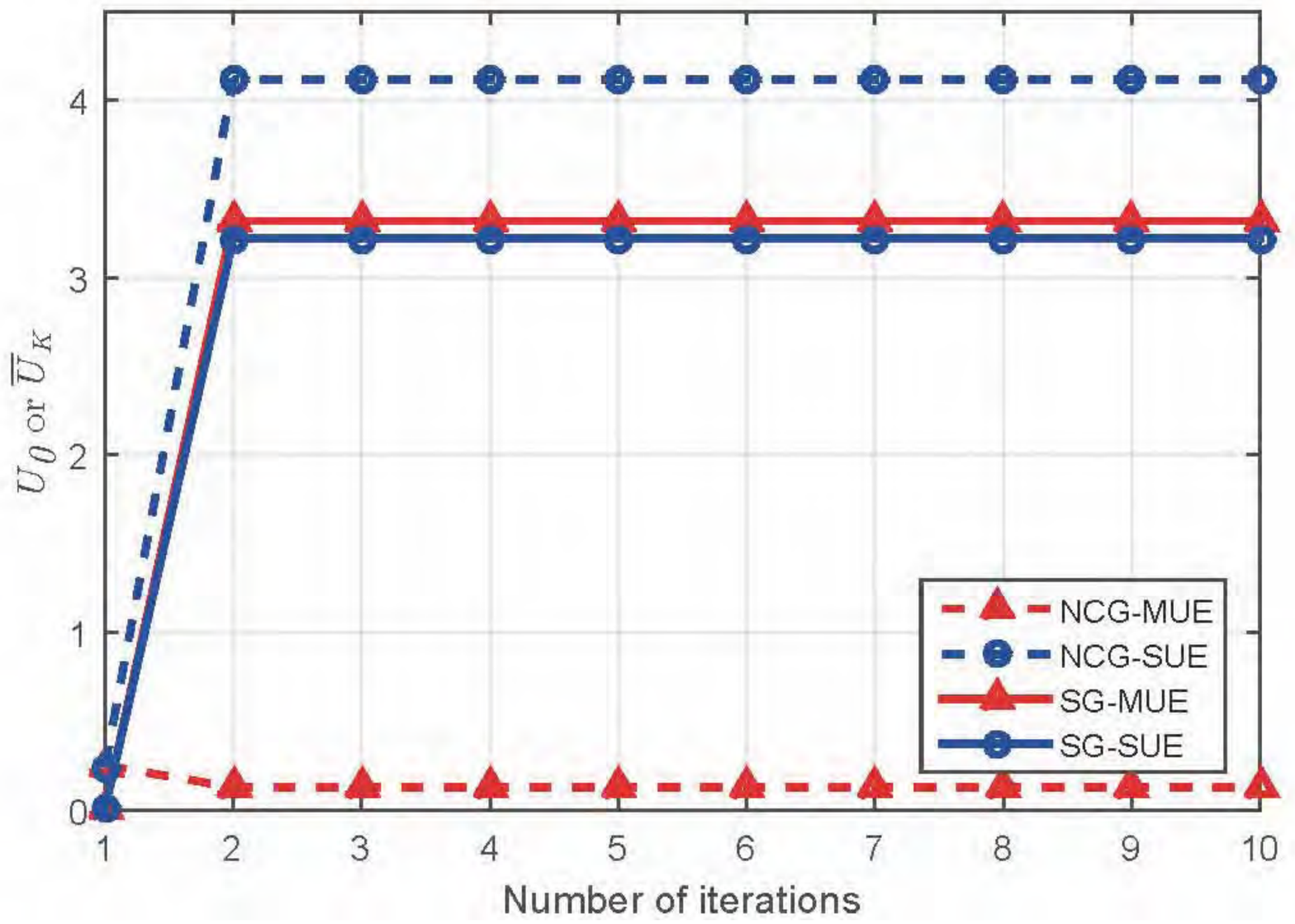}
\caption{The utility versus the number of iterations for different schemes.}
\label{fig4}
\end{figure}

In Fig. \ref{fig5} and Fig. \ref{fig6}, we show the transmit rate of the MUE and the average transmit rate of the SUEs of the proposed scheme versus $P_{T}$ with different $K$ for $\lambda=10^3$. As shown, the transmit rate first increases with  $P_{T}$ when $P_{T}$ is smaller than a certain threshold value, and then it approaches a steady value. When $P_{T}$ is sufficiently small, the MUE and the SUEs are constrained by the maximum transmit power. Correspondingly, their transmit rates are relatively small. As $P_{T}$ increases, their transmit rates increase due to the larger transmit power constraint. When $P_{T}$ is larger than a certain threshold value, the transmit power will increase, but it will also simultaneously cause more interference. Therefore, the transmit rate of the  MUE and that of the SUEs will both stop increasing.

\begin{figure}[!t]
\centering 
\includegraphics[width=0.53\textwidth]{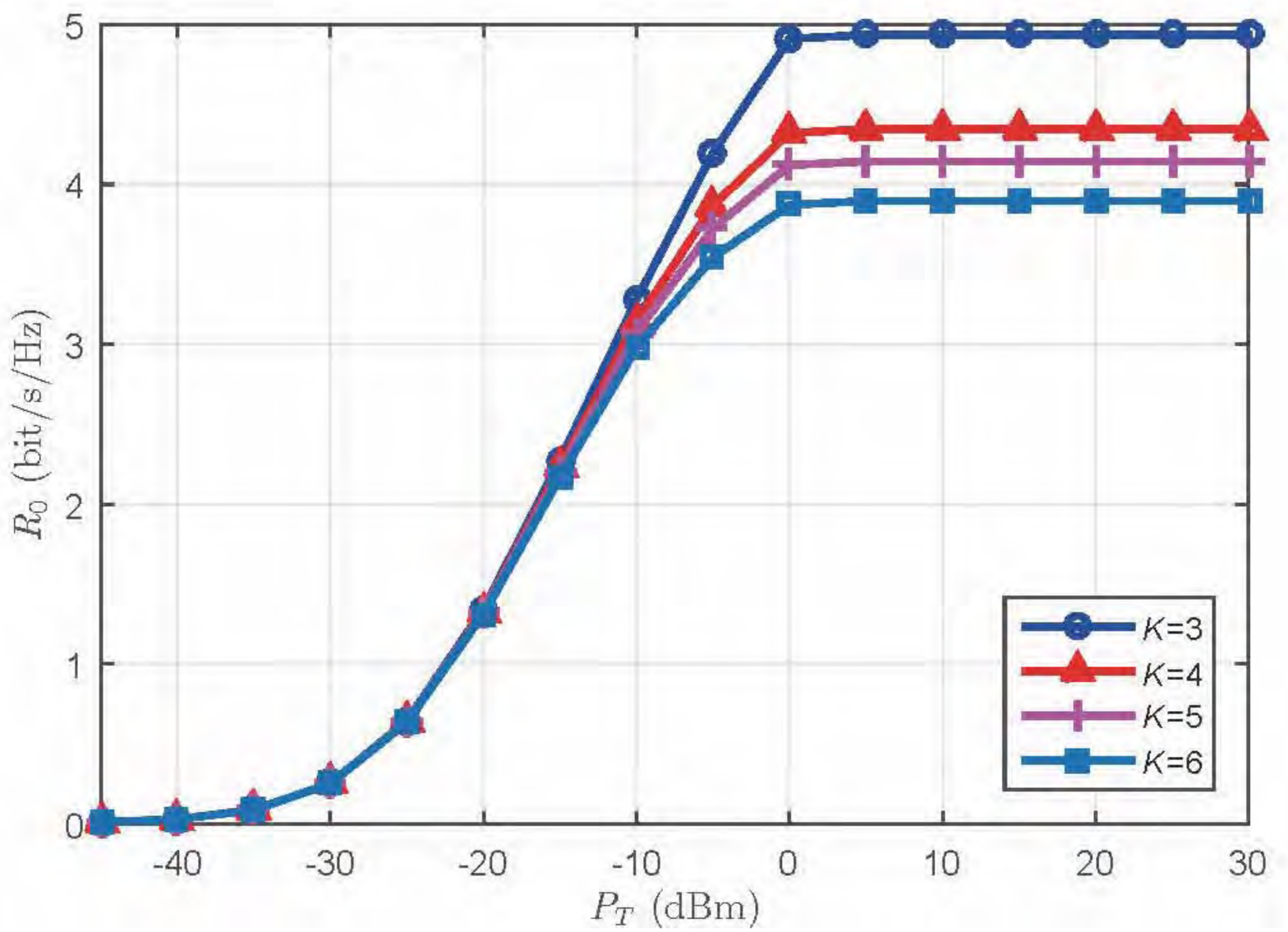}
\caption{The transmit rate of the MUE versus $P_{T}$ with different $K$.}
\label{fig5}
\end{figure}
\begin{figure}[!t]
\centering 
\includegraphics[width=0.53\textwidth]{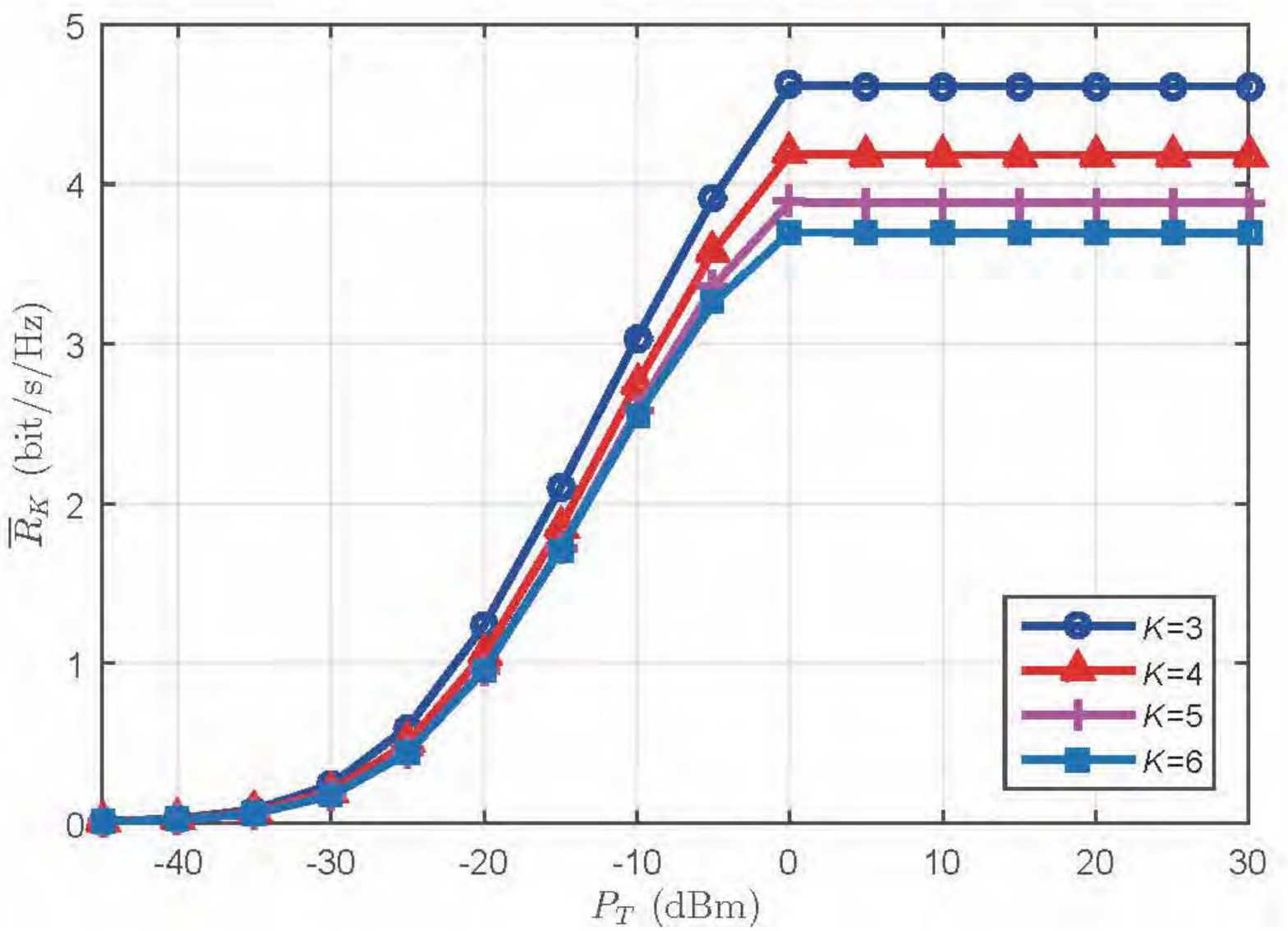}
\caption{The average transmit rate of the SUEs versus $P_{T}$ with different $K$.}
\label{fig6}
\end{figure}

In Fig. \ref{fig8}, we show the transmit rate of the MUE and the average transmit rate of the SUEs versus $P_{T}$ with different $\lambda_0$ and $\lambda_K$ for $K=4$. As shown,
the transmit rate of the MUE (the average transmit rate of the SUEs) is larger when $\lambda_0$ ($\lambda_K$) is relatively small.
The reason for this result is that a smaller cost of the transmit power will stimulate the corresponding player to employ a relatively large transmit power,
subsequently resulting in a larger transmit rate.

In Fig. \ref{fig7}, we show the average transmit rate of the MUE and SUEs of the proposed scheme versus $\lambda$ for $K=4$ and $P_T = 0$ dBm. As shown, the average transmit rate remains at a high value when $\lambda$ is smaller than 30 dB,   decreases gradually with the increased $\lambda$, and finally remains at a small value. The reason for this behavior is that the MUE or SUE will choose to decrease the transmit power and also the corresponding transmit rate
with the increased $\lambda$.

\begin{figure}[!t]
\centering 
\includegraphics[width=0.53\textwidth]{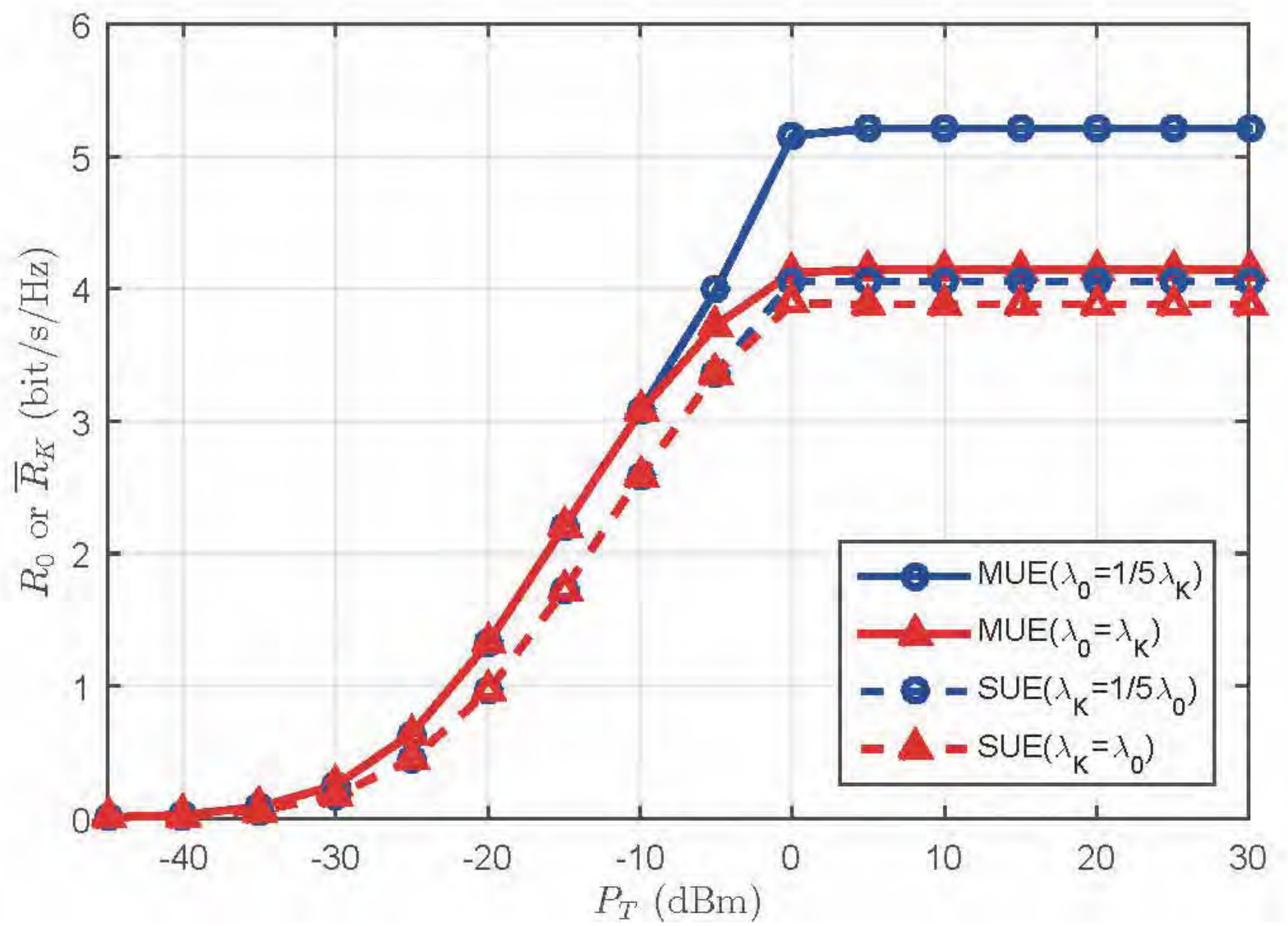}
\caption{The transmit rate versus $P_{T}$ with different $\lambda_0$ and $\lambda_K$.}
\label{fig8}
\end{figure}

\begin{figure}[!t]
\centering 
\includegraphics[width=0.53\textwidth]{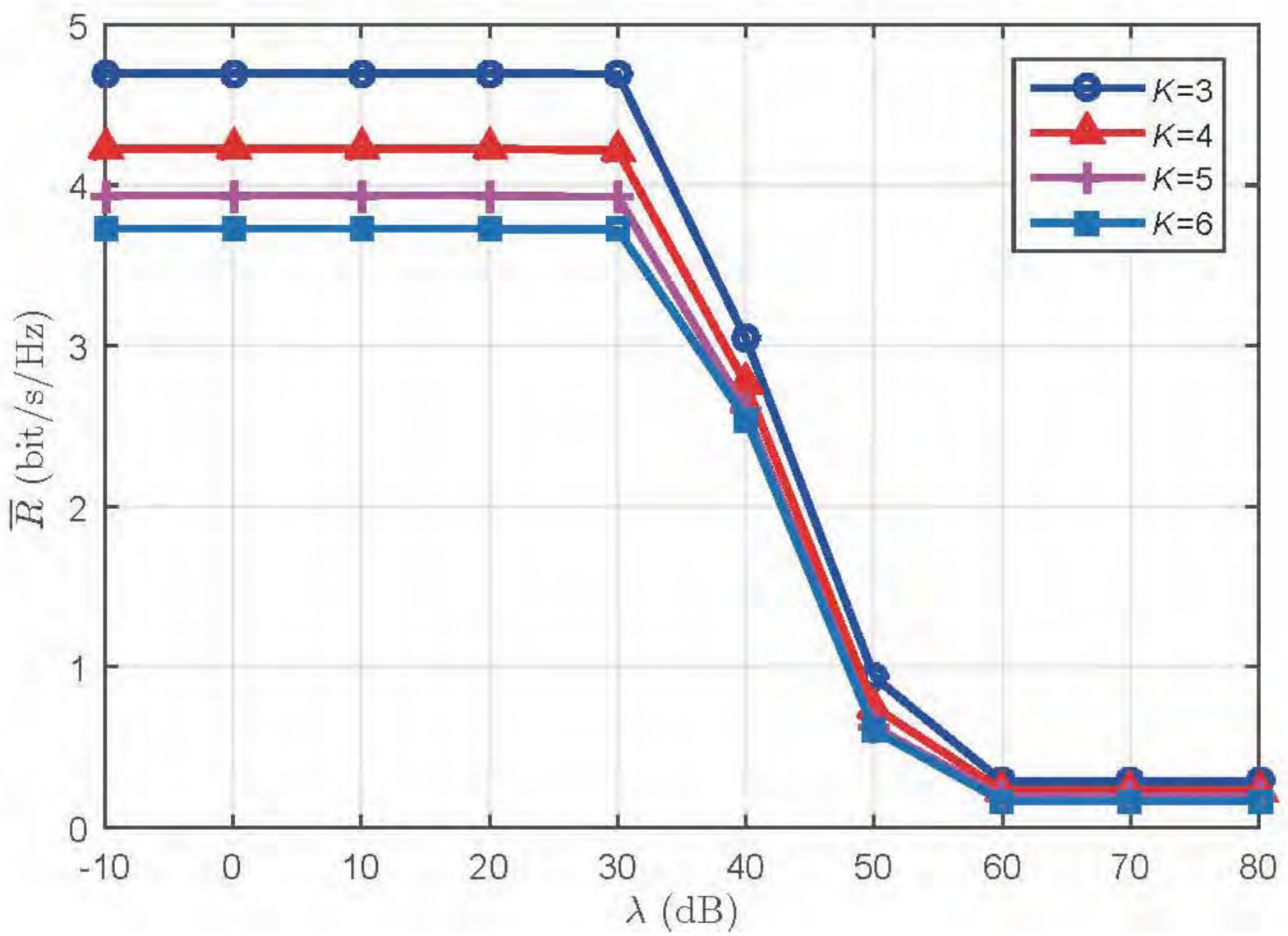}
\caption{The average transmit rate versus $\lambda$ with different $K$.}
\label{fig7}
\end{figure}

In Fig. \ref{fig9} and Fig. \ref{fig10}, we show the transmit power of the MUE and the average transmit power of the SUEs  versus $P_T$ with different $\lambda$ for $K=4$. As shown, the transmit power decreases with  $\lambda$. Moreover, the transmit power increases with  $P_T$ when $P_T$ is smaller than 0 dBm and remains approximately constant when $P_T$ is larger than 0 dBm. The reason for this result is that the maximum transmit power constraint will have no influence on the power control with a sufficiently large $P_T$.

\begin{figure}[!t]
\centering 
\includegraphics[width=0.52\textwidth]{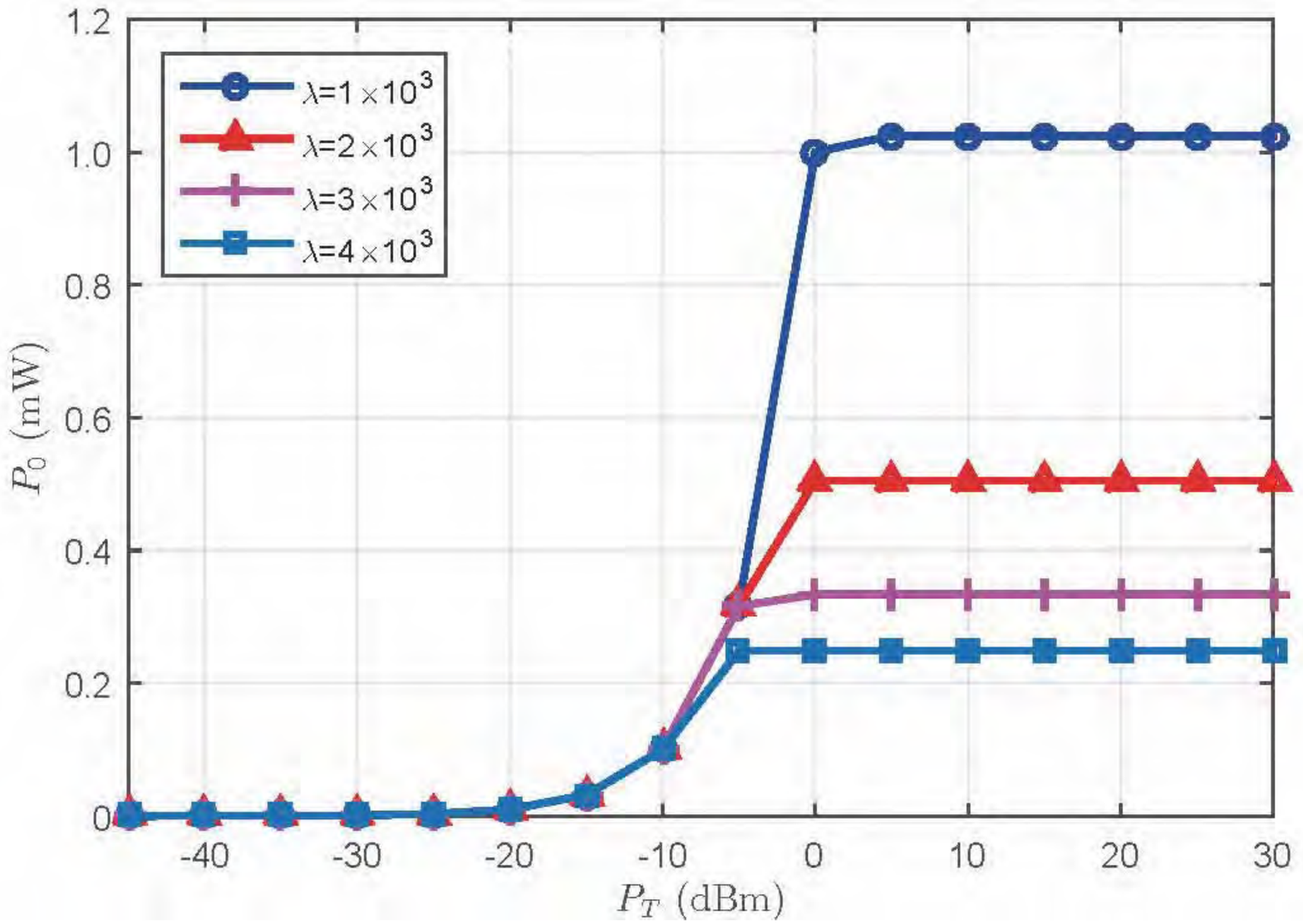}
\caption{The transmit power of the MUE versus $P_T$ with different $\lambda$.}
\label{fig9}
\end{figure}

\begin{figure}[!t]
\centering 
\includegraphics[width=0.52\textwidth]{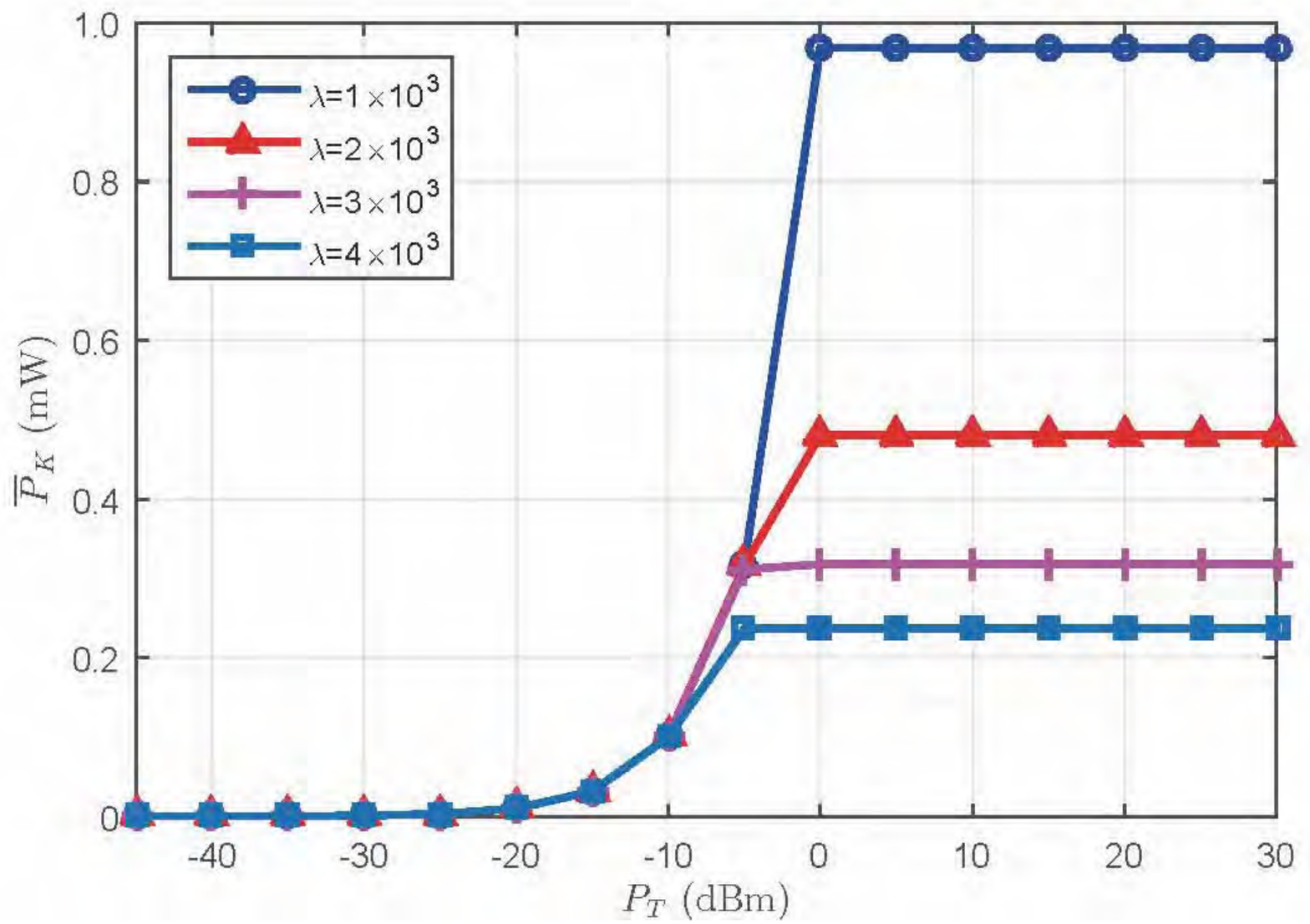}
\caption{The average transmit power of the SUEs versus $P_T$ with different $\lambda$.}
\label{fig10}
\end{figure}

\section{Conclusions}
In this paper, we have formulated a power control Stackelberg game for two-tier small-cell networks by considering both the transmit rate and cost.
The optimal transmit powers of the MUE and SUEs have been obtained based on the backward induction method.
We have developed a two-layer iterative power control scheme and  proven the convergence of this scheme. We have also shown the existence and uniqueness of the SE in the formulated Stackelberg game.
Numerical results have been presented to demonstrate the desirable performance of the proposed scheme.
For future work, we would like to explore power control with incomplete information for two-tier small-cell networks.

\section*{Data Availability}
The data used to support the findings of this study are included within the article.

\section*{Conflicts of Interest}
The authors declare that they have no competing interests.

\section*{Funding Statement}
This work was supported in part by
the Natural Science Foundation of China under Grant 61521061,
the Natural Science Foundation of Jiangsu Province under grant
BK20181264,
the Research Fund of the State Key Laboratory of
Integrated Services Networks (Xidian University) under grant ISN19-10,
the Research Fund of the Key Laboratory of Wireless Sensor Network $\&$ Communication (Shanghai Institute of Microsystem and Information Technology, Chinese Academy of Sciences) under grant 2017002,
the National Basic Research Program of China
(973 Program) under grant 2012CB316004,
and the U.K. Engineering and Physical Sciences Research Council under Grant EP/K040685/2.




\end{document}